\documentclass[11pt,draftcls, onecolumn]{IEEEtran}

\usepackage{amsmath}
\usepackage{amsthm}
\usepackage{amssymb}
\usepackage{verbatim}
\usepackage{algorithm}
\usepackage{algorithmic}
\usepackage{graphicx}
\usepackage{epstopdf}
\usepackage{paralist}
\usepackage{color}

\ifCLASSOPTIONcompsoc
\usepackage[tight,normalsize,sf,SF]{subfigure}
\else
\usepackage[tight,footnotesize]{subfigure}
\fi

\theoremstyle{plain}
\newtheorem{lemma}{Lemma}

\newtheorem{theorem}{Theorem}

\newtheorem{definition}{Definition}
\newtheorem{corollary}{Corollary}
\theoremstyle{definition}

\floatname{algorithm}{Algorithm}
\newtheorem*{remark}{Remark}
\newcommand{\argmax}{\operatornamewithlimits{argmax}}

\allowdisplaybreaks[4]

\begin{document}
%
\title{On Optimality of Myopic Sensing Policy with Imperfect Sensing in Multi-channel Opportunistic Access}
%
%
%
%

\author{Kehao~Wang \qquad Lin~Chen \qquad Quan~Liu \qquad Khaldoun~Al~Agha
\IEEEcompsocitemizethanks{\IEEEcompsocthanksitem K.~Wang, L.~Chen and K. Al~Agha are with the Laboratoire de Recherche en Informatique (LRI), Department of Computer Science, the University of Paris-Sud XI, 91405 Orsay, France (e-mail: \{Kehao.Wang, Lin.Chen, Khaldoun.Alagha\}@lri.fr). K.~Wang and Q.~Liu is with the school of Information Engineering, Wuhan University of Technology, 430070 Hubei, China (e-mail: \{Kehao.wang, Quan.Liu\}@whut.edu.cn).}}

\IEEEcompsoctitleabstractindextext{%
\begin{abstract}
We consider the channel access problem under imperfect sensing of channel state in a multi-channel opportunistic communication system, where the state of each channel evolves as an independent and identically distributed Markov process. The considered problem can be cast into a restless multi-armed bandit (RMAB) problem that is of fundamental importance in decision theory. It is well-known that solving the RMAB problem is PSPACE-hard, with the optimal policy usually intractable due to the exponential computation complexity. A natural alternative is to consider the easily implementable myopic policy that maximizes the immediate reward but ignores the impact of the current strategy on the future reward. In this paper, we perform an analytical study on the optimality of the myopic policy under imperfect sensing for the considered RMAB problem. Specifically, for a family of generic and practically important utility functions, we establish the closed-form conditions under which the myopic policy is guaranteed to be optimal even under imperfect sensing. Despite our focus on the opportunistic channel access, the obtained results are generic in nature and are widely applicable in a wide range of engineering domains.
\end{abstract}

\begin{IEEEkeywords}
Restless multi-armed bandit (RMAB) problem, myopic policy, imperfect sensing, opportunistic spectrum access (OSA)
\end{IEEEkeywords}}

\maketitle

\IEEEdisplaynotcompsoctitleabstractindextext

%
\IEEEpeerreviewmaketitle

\section{Introduction}
\label{section:introduction}

We consider an opportunistic multi-channel communication system in which a user has access to multiple channels, but is limited to sense and transmit only on a subset of them at a time. The fundamental problem we study is how the sender can exploit past observations and the knowledge of the stochastic properties of the channels to maximize its utility (e.g., expected throughput) by switching opportunistically across channels.

Formally, the considered channel access problem can be cast into the restless multi-armed bandit (RMAB) problem, one of the most well-known generalizations of the
classic multi-armed bandit (MAB) problem, which is of fundamental importance in stochastic decision theory. The standard formulation of the RMAB problem
can be briefly summarized as follows: There is a bandit of $N$ independent arms, each evolving as a
two-state Markov process. At each time slot, a player chooses $k$ ($1 \le k \le N$) of the $N$ arms to play and
receives a certain amount of reward depending on the state of the played arms. Given the initial state of
the system, the goal of the player is to find the optimal policy of playing the $k$ arms at each slot so as
to maximize the aggregated discounted long-term reward.

Despite the significant research efforts in the field, the RMAB problem in its generic form still remains
open. Until today, very little result is reported on the structure of the optimal policy. Obtaining the
optimal policy for a general RMAB problem is often intractable due to the exponential computation
complexity. Hence, a natural alternative is to seek a simple myopic policy maximizing the short-term
reward. Due to its simple and robust structure, the myopic sensing policy has begun to attract significant research attention, especially on the optimality of the myopic sensing policy.

The vast majority of studies in the area assume perfect observation of channel states. However, sensing or observation errors are inevitable in practical scenario (e.g., due to noise and system limitations), especially in wireless communication systems which is the focus of our work. More specifically, a good (bad, respectively) channel may be sensed as bad (good) and accessing a bad channel leads to zero reward. In such context, it is crucial to study the structure and the optimality of the myopic sensing policy with imperfect observation. We would like to emphasize that the presence of sensing error brings two difficulties when studying the myopic sensing policy in this new context.
\begin{itemize}
\item The channel state evolves as a non-linear mapping (w.r.t. the current channel state) instead of a linear one in the perfect sensing case.
\item In the non-perfect sensing case, the state transition of a channel depends not only on the channel evolution itself, but also on the observation outcome, meaning that the transition is not deterministic.
\end{itemize}

Due to the above particularities\footnote{Please refer to the remark of~\eqref{eq:belief_update_err_sym} for a detailed analysis}, our problem requires an original study on the optimality of the myopic sensing policy that cannot draw on existing results in the perfect sensing case. We would like to report that despite its practical importance and particularities, very few work has been done on the impact of sensing error on the performance of the myopic sensing policy, or more generically, on the RMAB problem under imperfect observation. To the best of our knowledge, \cite{Kliu10} is the only work in this area, where the optimality of the myopic policy is proved for the case of two channels with a particular utility function. In this paper, we derive closed-form conditions under which the myopic sensing policy is optimal under imperfect sensing for arbitrary $N$ and generic utility functions. As shown in Section~\ref{subsection:discussion}, the result obtained in this paper can cover the result of~\cite{Kliu10}.
Moreover, this paper also significantly extends our previous work~\cite{Wang11TPS}, focusing on perfect sensing scenario in which the analysis cannot be applied in the imperfect sensing scenario due to the non-trivial particularities introduced by sensing error as mentioned previously. In this regard, our work in this paper contributes the existing literature by developing an adapted analysis on the RMAB problem under imperfect sensing under the generic framework proposed in~\cite{Wang11TPS}.

The rest of the paper is organized as follows: Our model is formulated in Section~\ref{section:Problem_formulation}. Section~\ref{section:optimality} studies the optimality of the myopic sensing policy and illustrates the application of the derived results via two typical examples. A detailed discussion on the related work is given in Section~\ref{section:discussion}. Finally, the paper is concluded by Section~\ref{section:conclusion}.

\section{Problem Formulation}
\label{section:Problem_formulation}

\subsection{Multi-channel Opportunistic Access with Imperfect Sensing}

As outlined in the Introduction, we consider a multi-channel opportunistic communication system, in which a user is able to access a set $\cal N$ of $N$ independent and statistically identical channels, each characterized by a Markov chain of two states, \emph{good/idle} ($1$) and \emph{bad/busy} ($0$). The state transmission probabilities are given by $\{p_{i,j}\},i,j=0,1$. 
We assume that the system operates in a synchronously time slotted fashion with the time slot indexed by $t$ ($t = 1,2,\cdots,T$), where $T$ is the time horizon of interest. Each channel goes through state transition at the beginning of each slot $t$. This generic multi-channel opportunistic communication model can be naturally cast into the opportunistic spectrum access (OSA) problem in cognitive radio systems where an unlicensed secondary user can opportunistically access the temporarily unused channels of the licensed primary users, with the availability of each channel evolving as an independent Markov chain.

Limited by hardware constraints and energy cost, the user is allowed to sense only $k$ ($1\le k\le N$) of the $N$ channels at each slot $t$. We denote the set of channels chosen by the user at slot $t$ by ${\cal A}(t)$ where ${\cal A}(t)\in {\cal N}$ and $|{\cal A}(t)|=k$. We assume that the user makes the channel selection decision at the beginning of each slot after the channel state transition. Moreover, we are interested in the imperfect sensing scenario where channel sensing is subject to errors, i.e., a good channel may be sensed as bad one and vice versa. Let $\mathbf{S}(t)\triangleq[S_1(t),\cdots,S_N(t)]$ denote the channel state vector where $S_i(t)\in\{0,1\}$ is the state of channel $i$ in slot $t$ and let $\mathbf{S'}(t)\triangleq\{S'_i(t), i\in{A(t)}\}$ denote the sensing outcome vector where $S_i'(t)=0$ ($1$) means that the channel $i$ is sensed bad (good) in slot $t$. Using such notation, the performance of channel state detection is characterized by two system parameters: the probability of false alarm $\epsilon_i(t)$ and the probability of miss detection $\delta_i(t)$, formally defined as follows:
\begin{eqnarray*}
  \epsilon_i(t) \triangleq \text{Pr}\{S_i'(t)=1 | S_i(t)=0\}, \\
  \delta_i(t)   \triangleq \text{Pr}\{S_i'(t)=0 | S_i(t)=1\}.
\end{eqnarray*}
In our analysis, we consider the case where $\epsilon_i(t)$ and $\delta_i(t)$ are independent w.r.t. $t$ and $i$. More specifically, we defined $\epsilon$ and $\delta$ as the system-wide false alarm rate and miss detection rate. We also assume that when the receiver successfully receives a packet from a channel, it sends an acknowledgement to the transmitter
over the same channel at the end of the slot. The absence of an ACK signifies that the transmitter does not transmit
over this channel or transmitted but the channel is busy in this slot.

Obviously, by sensing only $k$ out of $N$ channels, the user cannot observe the state information of the whole system. Hence, the user has to infer the channel states from its past decision and observation history so as to make its future decision. To this end, we define the \emph{channel state belief vector} (hereinafter referred to as \emph{belief vector} for briefness) $\Omega(t)\triangleq\{\omega_i(t), i\in{\cal N}\}$, where $0\le \omega_i(t)\le 1$ is the conditional probability that channel $i$ is in state good (i.e., $S_i(t)=1$) at slot $t$ given all past states, actions and observations\footnote{The initial belief $\omega_i(1)$ can be set to $\frac{p_{01}}{p_{01}+1-p_{11}}$ if no information about the initial system state is available.}.
Due to the Markovian nature of the channel model, the belief vector can be updated recursively using Bayes Rule as shown in~\eqref{eq:belief_update_err_sym}.
\begin{equation}
\omega_i(t+1)=
\begin{cases}
p_{11}, & i\in {\cal A}(t), ACK=1 \\
\tau(\varphi(\omega_i(t))), & i\in {\cal A}(t), ACK=0 \\
\tau(\omega_i(t)), & i\not\in {\cal A}(t)
\end{cases},
\label{eq:belief_update_err_sym}
\end{equation}
where $ACK=1$ denotes the case where an ACK is received (successful transmission, i.e., $S_i'(i)=1$ and $S_i(t)=1$) and $ACK=0$ denotes the case where no ACK is received (failed transmission or no transmission, i.e., $S_i'(i)=1$ $S_i(t)=0$ or $S'(t)=0$),
$\varphi(\omega_i)=\frac{\epsilon \omega_{i}(t)}{\epsilon \omega_{i}(t) + 1 - \omega_{i}(t)}$ and
\begin{equation}
\tau(\omega_i(t))\triangleq\omega_i(t)p_{11}+[1-\omega_i(t)]p_{01}
\label{eq:tau_err_sym}
\end{equation}
denotes the operator for the one-step belief update.

\begin{remark}
We would like to emphasize that in contrast to the perfect sensing case~\cite{Wang11TPS} where $\omega_i(t+1)$ is a linear function of $\omega_i(t)$ whether $i$ in sensed or not, in the imperfect sensing case, the mapping from $\omega_i(t)$ to $\omega_i(t+1)$ is no longer linear due to the sensing error (cf. the second line of equation~\eqref{eq:belief_update_err_sym}). Moreover, the state transition of a channel depends not only on the channel evolution itself, but also on the observation outcome, i.e., $\omega_i(t+1)=p_{11}$ for $i\in {\cal A}(t), ACK=1$ and $\omega_i(t+1)=\tau(\varphi(\omega_i(t)))$ for $i\in {\cal A}(t), ACK=0$.
As will be shown later, these differences make the analysis for the imperfect sensing more complicated.
\end{remark}

To conclude this subsection, we state some structural properties of $\tau(\omega_i(t))$  and $\varphi(\omega_i(t))$ that are useful in the subsequent proofs.
\begin{lemma}
If $\forall i$, $p_{01}<p_{11}$, then
\begin{itemize}
\item $\tau(\omega_i(t))$ is monotonically increasing in $\omega_i(t)$;
\item $p_{01}\le \tau(\omega_i(t))\le p_{11}$, $\forall \ 0\le \omega_i(t)\le 1$.
\end{itemize}
\label{lemma:property_tau_err_sym}
\end{lemma}

\begin{proof}
Lemma~\ref{lemma:property_tau_err_sym} follows from $\tau(\omega_i(t))=(p_{11}-p_{01})\omega_i(t)+p_{01}$ straightforwardly.
\end{proof}

\begin{lemma}
If $0\leq \epsilon \leq \frac{(1-p_{11})p_{01}}{p_{11}(1-p_{01})}$, then
\begin{itemize}
  \item $\varphi(\omega_i(t))$ increases monotonically in $\omega_i(t)$ with $\varphi(0)=0$ and $\varphi(1)=1$;
  \item $\varphi(\omega_i(t)) \leq p_{01}$, $\forall p_{01}\leq \omega_i(t) \leq p_{11}$.
\end{itemize}
\label{lemma:property_phi_err_sym}
\end{lemma}
\begin{proof}
Noticing that $\varphi(\omega_i)=\frac{\epsilon \omega_{i}(t)}{\epsilon \omega_{i}(t) + 1 - \omega_{i}(t)}$, Lemma~\ref{lemma:property_phi_err_sym} follows straightforwardly.
\end{proof}

\subsection{Optimal Sensing Problem Formulation and Myopic Sensing Policy}

Given the imperfect sensing context, we are interested in the user's optimization problem to find the optimal sensing policy $\pi^*$ that maximizes the expected total discounted reward over a finite horizon. Mathematically, a sensing policy $\pi$ is defined as a mapping from the belief vector $\Omega(t)$ to the action (i.e., the set of channels to sense) ${\cal A}(t)$ in each slot $t$: $\pi: \ \Omega(t)\rightarrow{\cal A}(t), |{\cal A}(t)|=k, \ t=1, 2, \cdots, T.$

The following gives the formal definition of the optimal sensing problem:
\begin{equation}
\pi^*=\argmax_{\pi} \mathbb{E}\left.\left[\sum^{T}_{t=1} \beta^t R_{\pi}(\Omega(t))\right|\Omega(1)\right]
\label{eq:pb_formulation}
\end{equation}
where $R_{\pi}(\Omega(t))$ is the reward collected in slot $t$ under the sensing policy $\pi$ with the initial belief vector $\Omega(1)$, $0\le \beta\le 1$ is the discounting factor characterizing the feature that the future rewards are less valuable than the immediate reward. By treating the belief value of each channel as the state of each arm of a bandit, the user's optimization problem can be cast into a restless multi-armed bandit problem.

In order to get more insight on the structure of the optimization problem formulated in~\eqref{eq:pb_formulation} and the complexity to solve it, we derive the dynamic programming formulation of~\eqref{eq:pb_formulation} as follows:
\begin{align*}
V_{T}(\Omega(t))=& \max_{\pi} \mathbb{E}[R_{\pi}(\Omega(T))]=\max_{\substack{{\cal A}(T)\subseteq \mathcal{N} \\ |{\cal A}(T)|=k}} \mathbb{E}[R_{\pi}(\Omega(T))],  \\
V_{t}(\Omega(t))=& \max_{\substack{{\cal A}(t)\subseteq \mathcal{N} \\ |{\cal A}(t)|=k}} \mathbb{E}\left[ R_{\pi}(\Omega(t))+\beta  \sum_{{\mathcal E}\subseteq{\cal A}(t)} \prod_{i\in{\mathcal E}}(1-\epsilon)\omega_i(t) \right. \\ & \left. \prod_{j\in{\cal A}(t)\backslash{\mathcal E}}[1-(1-\epsilon)\omega_j(t)] V_{t+1}(\Omega(t+1))\right].
\end{align*}

In the above equations, $V_t(\Omega(t))$ is the value function corresponding to the maximal expected reward from time slot $t$ to $T$ ($1\le t\le T$) with the believe vector $\Omega(t+1)$ following the evolution described in~\eqref{eq:belief_update_err_sym} given that the channels in the subset $\cal E$ are sensed in state good and the channels in ${\cal A}(t)\backslash{\cal E}$ are sensed in state bad. 

Theoretically, the optimal policy can be obtained by solving the above dynamic programming. Unfortunately, due to the impact of the current action on the future reward and the unaccountable space of the belief vector, obtaining the optimal solution directly from the above recursive equations is computationally prohibitive. Hence, a natural alternative is to seek simple myopic sensing
policy which is easy to compute and implement that maximizes the expected immediate reward $F(\Omega(t))$, formally defined
as follows:
\begin{equation}
\label{eq:definition_myopic_p}
   \mathcal{A}(t)=\argmax_{\mathcal{A}(t)\subseteq{\mathcal{N}}} \Sigma_{i\in \mathcal{A}(t) } F(\Omega(t)).
\end{equation}

In this paper, we focus on a class of generic and practically important functions defined in~\cite{Wang11TPS} as \textit{regular} functions. More specifically, the expected immediate reward function $F(\Omega(t))$ studied in this paper are assumed to be symmetrical, monotonically non-decreasing and decomposable, defined by the three axioms in~\cite{Wang11TPS}.
Under this condition, the myopic policy consists of choosing the $k$ channels with the largest value of $\omega$. In the following sections we focus on the structure and the optimality of the myopic sensing policy under imperfect sensing. As pointed out in the remark following equations~\eqref{eq:belief_update_err_sym} and~\eqref{eq:tau_err_sym}, the main technical difficulties compared with the perfect sensing case are the non-linearity of the mapping from $\omega_i(t)$ to $\omega_i(t+1)$ and the dependency of the channel state transition on the observation outcome.

\section{Analysis on Optimality of Myopic Sensing Policy under Imperfect Sensing}
\label{section:optimality}

The goal of this section is to establish closed-form conditions under which the myopic sensing policy, despite of its simple structure, achieves the system optimum under imperfect sensing. To this end, we set up by defining an auxiliary function and studying the structural properties of the auxiliary function, which serve as a basis in the study of the optimality of the myopic sensing policy. We then establish the main result on the optimality followed by the illustration on how the obtained result can be applied via two concrete application examples.

For the convenience of discussion, we firstly state some notations before presenting the analysis:
\begin{itemize}
\item The believe vector $\Omega(t)$ is sorted to $[\omega_1(t), \cdots, \omega_N(t)]$ at each slot $t$ such that ${\cal A}= \{1, 2, \cdots, k\}$ {\footnote{For presentation simplicity, by slightly abusing the notations without introducing ambiguity, we drop the time slot index $t$.}};
\item ${\cal N}(m)\triangleq \{1,\cdots,m\} \ (m\le N)$ denotes the first $m$ channels in $\cal N$;
\item Given ${\cal E}\subseteq {\cal M}\subseteq{\cal N}$, $\displaystyle Pr({\cal M},{\mathcal E})\triangleq\prod_{i\in{\mathcal E}}(1-\epsilon)\omega_i(t)\prod_{j\in{\cal M}\setminus{\mathcal E}}[1-(1-\epsilon)\omega_j(t)]$, herein, $Pr({\cal M},{\mathcal E})$ denotes the expected probability that the channels in $\cal E$ are sensed in the good state, while the channels in ${\cal M}\setminus{\mathcal E}$ are sensed in the bad state, given that the channels in $\cal M$ are sensed;
\item $\mathbf{P_{11}^{\cal E}}$ denotes the vector of length $|{\cal E}|$ with each element being $p_{11}$;
\item $\mathbf{\Phi}(l,m)\triangleq[\tau(\omega_i(t)), l\le i\le m]$ where the components are sorted by channel index. $\mathbf{\Phi}(l,m)$ characterizes the updated belief values of the channels between $l$ and $m$ if they are not sensed;
\item Given ${\cal E}\subseteq {\cal M}\subseteq{\cal N}$, $\mathbf{Q^{{\cal M},{\cal E}}}\triangleq [\tau(\varphi(\omega_i(t))), i\in{\cal M}\setminus{\cal E}]$ where the components are sorted by channel index. $\mathbf{Q^{{\cal M},{\cal E}}}$ characterizes the updated belief values of the channels in ${\cal M}\setminus{\mathcal E}$ if they are sensed in the bad state; $\mathbf{\overline{Q}^{{\cal M},{\cal E},l}}\triangleq [\tau(\varphi(\omega_i(t))), i\in{\cal M}\setminus{\cal E} \text{ and } i<l]$ characterizes the updated belief values of the channels in ${\cal M}\setminus{\mathcal E}$ if they are sensed in the bad state with the channel index smaller than $l$; $\mathbf{\underline{Q}^{{\cal M},{\cal E},l}}\triangleq [\tau(\varphi(\omega_i(t))), i\in{\cal M}\setminus{\cal E} \text{ and } i>l]$ characterizes the updated belief values of the channels in ${\cal M}\setminus{\mathcal E}$ if they are sensed in the bad state with the channel index larger than $l$;
\item Let $\omega_{-i}\triangleq\{\omega_j, j\in{\cal A}, j\ne i\}$ and
    \begin{eqnarray*}
    \begin{cases}
    \displaystyle \Delta_{max}\triangleq \max_{ \omega_{-i}\in[0,1]^{k-1}} \ \{F(1, \omega_{-i})-F(0, \omega_{-i})\}, \\
    \displaystyle \Delta_{min}\triangleq \min_{ \omega_{-i}\in[0,1]^{k-1}} \ \{F(1, \omega_{-i})-F(0, \omega_{-i})\}.
    \end{cases}
    \end{eqnarray*}
\end{itemize}

\subsection{Definition and Properties of Auxiliary Value Function}

In this subsection, inspired by the form of the value function $V_t(\Omega(t))$ and the analysis in~\cite{Ahmad09}, we first define the auxiliary value function with imperfect sensing and then derive several fundamental properties of the auxiliary value function, which are crucial in the study on the optimality of the myopic sensing policy.

\begin{definition}[Auxiliary Value Function under Imperfect Sensing]
The auxiliary value function, denoted as $W_t(\Omega)$ ($t=1, 2, \cdots, T$) is recursively defined as follows:
\begin{align}
\label{eq:w_T_err_sym}
W_T(\Omega(T))=& F(\omega_1(T), \cdots, \omega_k(T)); \\
W_t(\Omega(t))=& F(\omega_1(t), \cdots, \omega_k(t))+  \nonumber \\
    & \beta \sum_{{\mathcal E}\subseteq{{\cal N}(k)}} Pr({\cal N}(k), {\mathcal E})W_{t+1}(\Omega_{{\mathcal E}}(t+1)),
\label{eq:w_t_err_sym}
\end{align}
where $\Omega_{{\mathcal E}}(t+1)\triangleq (\mathbf{P_{11}^{\mathcal E}}, \mathbf{\Phi}(k+1,N), \mathbf{Q^{{\cal N}(k),{\mathcal E}}})$ denotes the belief vector generated by $\Omega(t)$ based on~\eqref{eq:belief_update_err_sym}.
\end{definition}

The above recursively defined auxiliary value function gives the expected cumulated reward of the following sensing policy: in slot $t$, sense the first $k$ channels; if a channel $i$ is correctly sensed idle ($S_i'=1$ and $S_i=1$), then put it on the top of the list to be sensed in next slot, otherwise drop it to the bottom of the list. Recall Lemma~\ref{lemma:property_tau_err_sym} and Lemma~\ref{lemma:property_phi_err_sym}, under the condition $0\leq \epsilon \leq \frac{(1-p_{11})p_{01}}{p_{11}(1-p_{01})}$, if the belief vector $\Omega(t)$ is ordered decreasingly in slot $t$, the above sensing policy is the myopic sensing policy with $W_t(\Omega(t))$ being the total reward from slot $t$ to $T$.

In the subsequent analysis of this subsection, we prove some structural properties of the auxiliary value function.

\begin{lemma}[Symmetry]
\label{lemma:symmetry_sym_err}
If the expected reward function $F$ is regular, the correspondent auxiliary value function $W_t(\Omega)$ is symmetrical in any two channel $i, j\le k$ for all $t=1, 2, \cdots, T$, i.e.,
\begin{multline}
\label{lemma:avf_sym}
W_t(\omega_{1}, \cdots, \omega_i, \cdots, \omega_j, \cdots, \omega_N) = \\ W_t(\omega_{1}, \cdots, \omega_j, \cdots, \omega_i, \cdots, \omega_N), \quad \forall i, j\le k.
\end{multline}
\end{lemma}
\begin{proof}
The lemma can be easily shown by backward induction noticing that $(\omega_{1}, \cdots, \omega_i, \cdots, \omega_j, \cdots, \omega_N)$ and $(\omega_{1}, \cdots, \omega_j, \cdots, \omega_i, \cdots, \omega_N)$ generate the same belief vector $\Omega_{{\mathcal E}}(t+1)$ for any ${{\mathcal E}}$.
\end{proof}

\begin{lemma}[Decomposability]
\label{lemma:decomposability_sym_err}
If the expected reward function $F$ is regular, then the correspondent auxiliary value function $W_t(\Omega(t))$ is decomposable for all $t=1, 2, \cdots, T$, i.e.,
\begin{multline*}
W_t(\omega_{1}, \cdots, \omega_i, \cdots, \omega_N) =  \omega_i W_t(\omega_{1}, \cdots, 1, \cdots, \omega_N) + \\ (1-\omega_i)W_t(\omega_{1}, \cdots, 0, \cdots, \omega_N), \quad \forall i\in{\cal N}.
\end{multline*}
\end{lemma}

\begin{proof}
The proof is given in the appendix.
\end{proof}

Lemma~\ref{lemma:decomposability_sym_err} can be applied one step further to prove the following corollary.
\begin{corollary}
If the expected reward function $F$ is regular, then for any $l,m\in{\cal N}$ it holds that
\begin{multline*}
\label{corollary:decomposability_sym_err}
W_t(\omega_{1}, \cdots, \omega_l, \cdots, \omega_m, \cdots, \omega_N) - \\
W_t(\omega_{1}, \cdots, \omega_m, \cdots, \omega_l, \cdots, \omega_N) \\
= (\omega_l-\omega_m) \Big[W_t(\omega_{1}, \cdots, 1, \cdots, 0, \cdots, \omega_N) - \\
 W_t(\omega_{1}, \cdots, 0, \cdots,1, \cdots, \omega_N)\Big], \quad t=1, 2, \cdots, T.
\end{multline*}
\end{corollary}

\begin{lemma}[Monotonicity]
\label{lemma:monotonicity_sym_err}
If the expected reward function $F$ is regular, the correspondent auxiliary value function $W_t(\Omega)$ is monotonously non-decreasing in $\omega_l$, $\forall l\in{\cal N}$, i.e.,
$$\omega_l'\ge \omega_l \Longrightarrow W_t(\omega_1, \cdots, \omega_l', \cdots, \omega_N)\ge W_t(\omega_1, \cdots, \omega_l, \cdots, \omega_N).$$
\end{lemma}
\begin{proof}
The proof is given in the appendix.
\end{proof}

\subsection{Optimality of Myopic Sensing under Imperfect Sensing}

In this section, we study the optimality of the myopic sensing policy under imperfect sensing. We start by showing the following important auxiliary lemmas (Lemma~\ref{lemma:exchange_err_sym} and Lemma~\ref{lemma:upper_bound_err_sym}) and then establish the sufficient condition under which the optimality of the myopic sensing policy is guaranteed.

\begin{lemma}
\label{lemma:exchange_err_sym}
Given that (1) $F$ is regular, (2) $\epsilon <\frac{p_{01}(1-p_{11})}{P_{11}(1-p_{01})}$, and (3) $\beta \le \frac{\Delta_{min}}{\Delta_{max}\left[(1-\epsilon)(1-p_{01}) + \frac{\epsilon(p_{11}-p_{01})}{1-(1-\epsilon)(p_{11}-p_{01})}\right]}$, if $p_{11} \ge \omega_{l}\ge \omega_m \ge p_{01}$ where $l<m$, then it holds that
\begin{multline*}
W_t(\omega_{1},\cdots,\omega_{l},\cdots,\omega_{m},\cdots,\omega_{N})\ge \\ W_{t}(\omega_{1},\cdots,\omega_{m},\cdots,\omega_{l},\cdots,\omega_{N}), \quad t=1,\cdots,T.
\end{multline*}
\end{lemma}

\begin{lemma}
\label{lemma:upper_bound_err_sym}
Given that (1) $F$ is regular, (2) $\epsilon <\frac{p_{01}(1-p_{11})}{P_{11}(1-p_{01})}$, and (3) $\beta \le \frac{\Delta_{min}}{\Delta_{max}\left[(1-\epsilon)(1-p_{01}) + \frac{\epsilon(p_{11}-p_{01})}{1-(1-\epsilon)(p_{11}-p_{01})}\right]}$, if $p_{11}\ge \omega_1\ge \cdots \ge \omega_N\ge p_{01}$, for any $1\le t\le T$, it holds that

\begin{multline*}
W_t(\omega_{1},\cdots,\omega_{N-1},\omega_N)-W_{t}(\omega_{N},\omega_{1},\cdots,\omega_{N-1}) \le (1-\omega_{N})\Delta_{max},
\end{multline*}
\begin{multline*}
W_t(\omega_{1},\omega_{2},\cdots,\omega_{N-1},\omega_N)-W_{t}(\omega_{N},\omega_{2},\cdots,\omega_{N-1}, \omega_1) \le \\ (p_{11}-p_{01})\Delta_{max}\frac{1-[\beta(1-\epsilon)(p_{11}-p_{01})]^{T-t+1}}{1-\beta(1-\epsilon)(p_{11}-p_{01})}.
\end{multline*}
\end{lemma}

Lemma~\ref{lemma:exchange_err_sym} states that by swapping two elements in $\Omega$ with the former larger than the latter, the user does not increase the total expected reward. Lemma~\ref{lemma:upper_bound_err_sym}, on the other hand, gives the upper bound on the difference of the total reward of the two swapping operations, swapping $\omega_N$ and $\omega_k$ ($k=N-1,\cdots,1$) and swapping $\omega_1$ and $\omega_N$, respectively. For clarity of presentation, the detailed proofs of the two lemmas are deferred to the Appendix. From a technical point of view, it is insightful to compare the methodology in the proof with that in the analysis presented in~\cite{Ahmad09b} for the perfect sensing case with $k=1$. The key point of the analysis in~\cite{Ahmad09b} lies in the coupling argument leading to Lemma~3 in~\cite{Ahmad09b}. This analysis, however, cannot be directly applied in the generic case with imperfect sensing due to the non-linearity of the belief vector update as stated in the remark after equation \eqref{eq:belief_update_err_sym}. Hence, we base our analysis on the intrinsic structure of the auxiliary value function $W$ and investigate the different "branches" of channel realizations to derive the relevant bounds, which are further applied to study the optimality of the myopic sensing policy, as stated in the following theorem.

\begin{theorem}
\label{theorem:optimal_condition_err_sym}
If ${p_{01}}\leq{\omega_{i}(1)}\leq{p_{11}}, {1}\leq{i}\leq{N}$, the myopic sensing policy is optimal if the following conditions hold: (1) $F(\Omega)$ is regular; (2) $\epsilon <\frac{p_{01}(1-p_{11})}{P_{11}(1-p_{01})}$; (3) $\beta \le \frac{\Delta_{min}}{\Delta_{max}\left[(1-\epsilon)(1-p_{01}) + \frac{\epsilon(p_{11}-p_{01})}{1-(1-\epsilon)(p_{11}-p_{01})}\right]}$.
\end{theorem}

\begin{proof}
It suffices to show that for $t=1, \cdots, T$, by sorting $\Omega(t)$ in decreasing order such that $\omega_1\ge\cdots\ge\omega_N$, it holds that $W_t(\omega_1,\cdots,\omega_N)\ge W_t(\omega_{i_1},\cdots,\omega_{i_N})$, where $(\omega_{i_1},\cdots,\omega_{i_N})$ is any permutation of $(1,\cdots,N)$.

We prove the above inequality by contradiction. Assume, by contradiction, the maximum of $W_t$ is achieved at $(\omega_{i_1^*},\cdots,\omega_{i_N^*})\ne (\omega_1,\cdots,\omega_N)$, i.e.,
\begin{equation}
W_t(\omega_{i_1^*},\cdots,\omega_{i_N^*})>W_t(\omega_1,\cdots,\omega_N).
\label{eq:aux1}
\end{equation}

However, run a bubble sort algorithm on $(\omega_{i_1^*},\cdots,\omega_{i_N^*})$ by repeatedly stepping through it, comparing each pair of adjacent element $\omega_{i_l^*}$ and $\omega_{i_{l+1}^*}$ and swapping them if $\omega_{i_l^*}<\omega_{i_l^*+1}$. Note that when the algorithm terminates, the channel belief vector are sorted decreasingly, that is to say, it becomes $(\omega_1,\cdots,\omega_N)$.  By applying Lemma~\ref{lemma:exchange_err_sym} at each swapping, we have $W_t(\omega_{i_1^*},\cdots,\omega_{i_N^*})\le W_t(\omega_1,\cdots,\omega_N)$, which contradicts to~\eqref{eq:aux1}.  Theorem~\ref{theorem:optimal_condition_err_sym} is thus proven.
\end{proof}

As noted in~\cite{Kliu10}, when the initial belief $\omega_i$ is set to $\frac{p_{01}}{p_{01}+1-p_{11}}$ as is often the case in practical systems, it can be checked that ${p_{01}}\leq{\omega_{i}(1)}\leq{p_{11}}$ holds. Moreover, even the initial belief does not fall in $[p_{01},p_{11}]$, all the the belief values are bounded in the interval from the second slot following Lemma~\ref{lemma:property_tau_err_sym}. Hence our results can be extended by treating the first slot separately from the future slots.

\subsection{Discussion}
\label{subsection:discussion}

In this subsection, we illustrate the application of the result obtained above in two concrete scenarios and compare our work with the existing results.

Consider the channel access problem in which the user is limited to sense $k$ channels and gets one unit of reward if a sensed channel is in the good state, i.e., the utility function can be formulated as $F(\Omega_A)=(1-\epsilon)\sum_{i\in A}\omega_i$. Note that the optimality of the myopic sensing policy under this model is studied in~\cite{Kliu10} for a subset of scenarios where $k=1$, $N=2$. We now study the generic case with $k, N\ge 2$. To that end, we apply Theorem~\ref{theorem:optimal_condition_err_sym}. Notice in this example, we have $\Delta_{min}=\Delta_{max}=1-\epsilon$. We can then verify that when $\epsilon <\frac{p_{01}(1-p_{11})}{P_{11}(1-p_{10})}$, it holds that $\frac{\Delta_{min}}{\Delta_{max}[(1-\epsilon)(1-p_{01}) + \frac{\epsilon(p_{11}-p_{01})}{1-(1-\epsilon)(p_{11}-p_{01})}]}>1$. Therefore, when the condition 1 and 2 holds, the myopic sensing policy is optimal for any $\beta$. This result in generic cases significantly extends the results obtained in~\cite{Kliu10} where the optimality of the myopic policy is proved for the case of two channels and only conjectured for general cases.

Next consider another scenario where the user can sense $k$ channels but can only choose one of them to transmit its packets. Under this model, the user wants to maximize its expected throughput. More specifically, the slot utility function $F=F(\Omega_A)= 1- \Pi_{i\in{\cal A}} [1-(1-\epsilon)\omega_i]$, which is regular. In this context, we have $\Delta_{max}=(1-\epsilon)^{k-1}p_{11}^{k-1}$ and $\Delta_{min}=(1-\epsilon)^{k-1}p_{01}^{k-1}$. The third condition on for the myopic policy to be optimal becomes $\beta \le \frac{p_{01}^{k-1}}{ p_{11}^{k-1} [{(1-\epsilon)(1-p_{01}) + \frac{\epsilon(p_{11}-p_{01})}{1-(1-\epsilon)(p_{11}-p_{01})}}]}$. Particularly, when $\epsilon=0$, $\beta\le\frac{p_{01}^{k-1}}{p_{11}^{k-1}(1-p_{01})}$. It can be noted that even when there is no sensing error, the myopic policy is not ensured to be optimal, which confirms our findings in previous work \cite{Wang11} on perfect sensing scenarios.

\section{Related Work}
\label{section:discussion}

Due to its application in numerous engineering problems, the restless multi-armed bandit (RMAB) problem is of fundamental importance in stochastic decision theory. However, finding the optimal policy in the generic RMAB problem is shown to be PSPACE-hard by Papadimitriou~\emph{et al.} in~\cite{Papadimitriou99}. Whittle proposed a heuristic index policy, called Whittle index policy~\cite{Whittle88} which are shown to be asymptotically optimal in certain limited regime under some specific constraints~\cite{Weber90}. Unfortunately, not every RMAB problem has a well-defined Whittle index. Moreover, computing the Whittle index can be prohibitively complex. In this regard, Liu~\emph{et al.} studied in~\cite{Liu10} the indexability of a class of RMAB problems relevant to dynamic multi-channel access applications. However, the optimality of the myopic policy based on Whittle index is not ensured in the general cases, especially when the arms follow non-identical Markov chains.

A natural alternative, given that the RMAB problem is not tractable, is to seek simple myopic policies maximizing the short-term reward. In this line of research, significant research efforts have been devoted to studying the performance gap between the myopic policy and the optimal one and designing approximation algorithms and heuristic policies (cf. \cite{Guha07,Guha09,Bertsimas00}). Specifically, a simple myopic policy, termed as greedy policy, is developed in~\cite{Guha07} that yields a factor $2$ approximation of the optimal policy for a subclass of scenarios referred to as \emph{Monotone bandits}. Recently, the RMAB problem finds its application in the opportunistic channel access and has motivated the study of the myopic sensing policy in this context. More specifically, the structure of the myopic sensing policy is studied in~\cite{Qzhao08}. The optimality of the myopic sensing policy is derived in~\cite{Ahmad09b} for the positively correlated channels when the sender is limited to choose one channel each time (i.e., $k=1$). The result is further extended in to the case of sensing multiple channels ($k\ge 1$) channels in~\cite{Ahmad09} for a particular form of utility function modeling the fact that the user gets one unit of reward for each channel sensed good. A separation principle has been established in [11] which reveals the optimality of the myopic approach in the design of the channel state detector and the access policy. Our previous work~\cite{Wang11TPS}~\cite{Wang11TIS} adopts another line of research by focusing a family of generic and practically important utility functions and deriving closed-form conditions under which the myopic sensing policy is ensured to be optimal. In the context of imperfect sensing, the optimality of the myopic sensing policy is proved for the case of $N=2$ and $k=1$ in~\cite{Kliu10}. Our work presented in this paper contributes the literature by deriving the closed-form conditions on the optimality of the myopic sensing policy with imperfect sensing in the general case.

\section{Conclusion}
\label{section:conclusion}

In this paper, we have investigated the problem of opportunistic channel access under imperfect channel state sensing. We have derived closed-form conditions under which the myopic sensing policy is ensured to be optimal. Due to the generic RMAB formulation of the problem, the obtained results and the analysis methodology presented in this paper are widely applicable in a wide range of domains.

\appendices
\section{Proof of Lemma~\ref{lemma:decomposability_sym_err}}
We proceed the proof by backward induction. Firstly, it is easy to verify that the lemma holds for slot $T$.

Assume that the lemma holds from slots $t+1,\cdots,T$, we now prove it also holds for slot $t$ by the following two different cases.
\begin{itemize}
  \item Case 1: channel $l$ is not sensed in slot $t$, i.e. $l\ge k+1$. Let ${\cal M}\triangleq {\cal N}(k)=\{1,\cdots,k\}$, $\omega_l= 0$ and $1$, respectively, we have
\begin{eqnarray*}
  W_t(\omega_1,\cdots,\omega_l,\cdots,\omega_n) &=& F(\omega_1,\cdots,\omega_k)+\beta \sum_{{\cal E}\subseteq{\cal M}}Pr({\cal M}, {\cal E}) W_{t+1}(\Omega_{l}^{{\mathcal E}}(t+1)), \\
  W_t(\omega_1,\cdots,0,\cdots,\omega_n) &=& F(\omega_1,\cdots,\omega_k)+\beta \sum_{{\cal E}\subseteq{\cal M}}Pr({\cal M}, {\cal E}) W_{t+1}(\Omega_{l,0}^{{\mathcal E}}(t+1)), \\
  W_t(\omega_1,\cdots,1,\cdots,\omega_n) &=& F(\omega_1,\cdots,\omega_k)+\beta \sum_{{\cal E}\subseteq{\cal M}}Pr({\cal M}, {\cal E}) W_{t+1}(\Omega_{l,1}^{{\mathcal E}}(t+1)),
\end{eqnarray*}
where
\begin{eqnarray*}
\Omega_{l}^{{\mathcal E}}(t+1)&=&(\mathbf{P_{11}^{\mathcal E}}, \mathbf{\Phi}(k+1,l-1),\tau(\omega_l),\mathbf{\Phi}(l+1,N),\mathbf{{Q}^{{\cal M},{\mathcal E}}}), \\
\Omega_{l,0}^{{\mathcal E}}(t+1)&=&(\mathbf{P_{11}^{\mathcal E}}, \mathbf{\Phi}(k+1,l-1),p_{01},\mathbf{\Phi}(l+1,N),\mathbf{{Q}^{{\cal M},{\mathcal E}}}), \\
\Omega_{l,1}^{{\mathcal E}}(t+1)&=&(\mathbf{P_{11}^{\mathcal E}}, \mathbf{\Phi}(k+1,l-1),p_{11},\mathbf{\Phi}(l+1,N), \mathbf{{Q}^{{\cal M},{\mathcal E}}}).
\end{eqnarray*}
To prove the lemma in this case, it is sufficient to prove
\begin{equation}
\label{no_Acc_decomposability_sym_err}
    W_{t+1}(\Omega_{l}^{{\mathcal E}}(t+1)) = (1-\omega_l)W_{t+1}(\Omega_{l,0}^{{\mathcal E}}(t+1)) + \omega_l W_{t+1}(\Omega_{l,1}^{{\mathcal E}}(t+1))
\end{equation}
According to induction result, we have
\begin{eqnarray}
\label{no_Acc_decomposability_sym_err_case1}
\begin{split}
  W_{t+1}(\Omega_{l}^{{\mathcal E}}(t+1)) = & \tau(\omega_l)\cdot W_{t+1}(\mathbf{P_{11}^{\mathcal E}}, \mathbf{\Phi}(k+1,l-1),1,\mathbf{\Phi}(l+1,N),\mathbf{{Q}^{{\cal M},{\mathcal E}}}) \\
  & + (1-\tau(\omega_l))\cdot W_{t+1}(\mathbf{P_{11}^{\mathcal E}}, \mathbf{\Phi}(k+1,l-1),0,\mathbf{\Phi}(l+1,N),\mathbf{{Q}^{{\cal M},{\mathcal E}}})
\end{split}
\end{eqnarray}
\begin{eqnarray}
\label{no_Acc_decomposability_sym_err_case2}
\begin{split}
  W_{t+1}(\Omega_{l,0}^{{\mathcal E}}(t+1)) = &  p_{01}\cdot W_{t+1}(\mathbf{P_{11}^{\mathcal E}}, \mathbf{\Phi}(k+1,l-1),1,\mathbf{\Phi}(l+1,N),\mathbf{{Q}^{{\cal M},{\mathcal E}}}) \\
  & + (1-p_{01})\cdot W_{t+1}(\mathbf{P_{11}^{\mathcal E}}, \mathbf{\Phi}(k+1,l-1),0,\mathbf{\Phi}(l+1,N),\mathbf{{Q}^{{\cal M},{\mathcal E}}})
\end{split}
\end{eqnarray}
\begin{eqnarray}
\label{no_Acc_decomposability_sym_err_case3}
\begin{split}
   W_{t+1}(\Omega_{l,0}^{{\mathcal E}}(t+1)) = &  p_{11}\cdot W_{t+1}(\mathbf{P_{11}^{\mathcal E}}, \mathbf{\Phi}(k+1,l-1),1,\mathbf{\Phi}(l+1,N),\mathbf{{Q}^{{\cal M},{\mathcal E}}}) \\
  & + (1-p_{11})\cdot W_{t+1}(\mathbf{P_{11}^{\mathcal E}}, \mathbf{\Phi}(k+1,l-1),0,\mathbf{\Phi}(l+1,N),\mathbf{{Q}^{{\cal M},{\mathcal E}}})
\end{split}
\end{eqnarray}
Combing~\eqref{no_Acc_decomposability_sym_err_case1},~\eqref{no_Acc_decomposability_sym_err_case2},~\eqref{no_Acc_decomposability_sym_err_case3}, we have~\eqref{no_Acc_decomposability_sym_err}.

  \item Case 2: channel $l$ is sensed in slot $t$, i.e. $l\le k$. Let ${\cal M}\triangleq {\cal N}(k)\setminus\{l\}=\{1, \cdots,l-1,l+1,\cdots,k\}$, we have according to~\eqref{eq:w_t_err_sym}
  \begin{equation*}
    \begin{split}
       W_t(\Omega(t))= & F(\omega_1, \cdots,\omega_l,\cdots, \omega_k)\\
       & +\beta (1-\epsilon) \omega_l \sum_{{\cal E}\subseteq{\cal M}}Pr({\cal M}, {\cal E}) W_{t+1}(\mathbf{P_{11}^{\mathcal E}}, p_{11}, \mathbf{\Phi}(k+1,N), \mathbf{\overline{Q}^{{\cal M},{\mathcal E},l}},\mathbf{\underline{Q}^{{\cal M},{\mathcal E},l}}) \\
       & +\beta [1-(1-\epsilon)\omega_l] \sum_{{\cal E}\subseteq{\cal M}}Pr({\cal M}, {\cal E}) W_{t+1}(\mathbf{P_{11}^{\mathcal E}}, \mathbf{\Phi}(k+1,N), \mathbf{\overline{Q}^{{\cal M},{\mathcal E},l}},\tau(\varphi(\omega_l)), \mathbf{\underline{Q}^{{\cal M},{\mathcal E},l}})
    \end{split}
  \end{equation*}

  Let $\omega_l=0$ and $1$, respectively, we have
\begin{equation*}
\begin{split}
  W_t(\omega_1,\cdots,0,\cdots,\omega_n) =& F(\omega_1,\cdots,0,\cdots,\omega_k) \\
  &+\beta \sum_{{\cal E}\subseteq{\cal M}}Pr({\cal M}, {\cal E}) W_{t+1}(\mathbf{P_{11}^{\mathcal E}}, \mathbf{\Phi}(k+1,N), \mathbf{\overline{Q}^{{\cal M},{\mathcal E},l}},p_{01}, \mathbf{\underline{Q}^{{\cal M},{\mathcal E},l}}),
\end{split}
\end{equation*}

\begin{equation*}
\begin{split}
  W_t(\omega_1,\cdots,1,\cdots,\omega_n) =& F(\omega_1,\cdots,1,\cdots,\omega_k) \\
  &+\beta (1-\epsilon)\sum_{{\cal E}\subseteq{\cal M}}Pr({\cal M}, {\cal E}) W_{t+1}(\mathbf{P_{11}^{\mathcal E}}, p_{11}, \mathbf{\Phi}(k+1,N), \mathbf{\overline{Q}^{{\cal M},{\mathcal E},l}},\mathbf{\underline{Q}^{{\cal M},{\mathcal E},l}})\\
  &+\beta \epsilon\sum_{{\cal E}\subseteq{\cal M}}Pr({\cal M}, {\cal E}) W_{t+1}(\mathbf{P_{11}^{\mathcal E}}, \mathbf{\Phi}(k+1,N), \mathbf{\overline{Q}^{{\cal M},{\mathcal E},l}},p_{11}, \mathbf{\underline{Q}^{{\cal M},{\mathcal E},l}})
\end{split}
\end{equation*}
To prove the lemma in this case, it is sufficient to show
\begin{multline}
\label{Acc_decomposability_sym_err}
   [1-(1-\epsilon)\omega_l] W_{t+1}(\mathbf{P_{11}^{\mathcal E}}, \mathbf{\Phi}(k+1,N), \mathbf{\overline{Q}^{{\cal M},{\mathcal E},l}},\tau(\varphi(\omega_l)), \mathbf{\underline{Q}^{{\cal M},{\mathcal E},l}}) \\
   = (1-\omega_l) W_{t+1}(\mathbf{P_{11}^{\mathcal E}}, \mathbf{\Phi}(k+1,N), \mathbf{\overline{Q}^{{\cal M},{\mathcal E},l}},p_{01},\mathbf{\underline{Q}^{{\cal M},{\mathcal E},l}}) \\
   + \epsilon \omega_l  W_{t+1}(\mathbf{P_{11}^{\mathcal E}}, \mathbf{\Phi}(k+1,N), \mathbf{\overline{Q}^{{\cal M},{\mathcal E},l}},p_{11}, \mathbf{\underline{Q}^{{\cal M},{\mathcal E},l}})
\end{multline}

According to induction result, we have
\begin{multline}
\label{Acc_decomposability_sym_err_case1}
     W_{t+1}(\mathbf{P_{11}^{\mathcal E}}, \mathbf{\Phi}(k+1,N), \mathbf{\overline{Q}^{{\cal M},{\mathcal E},l}},\tau(\varphi(\omega_l)), \mathbf{\underline{Q}^{{\cal M},{\mathcal E},l}}) \\
     =\tau(\varphi(\omega_l)) W_{t+1}(\mathbf{P_{11}^{\mathcal E}}, \mathbf{\Phi}(k+1,N), \mathbf{\overline{Q}^{{\cal M},{\mathcal E},l}},1, \mathbf{\underline{Q}^{{\cal M},{\mathcal E},l}}) \\
     + (1-\tau(\varphi(\omega_l))) W_{t+1}(\mathbf{P_{11}^{\mathcal E}}, \mathbf{\Phi}(k+1,N), \mathbf{\overline{Q}^{{\cal M},{\mathcal E},l}},0, \mathbf{\underline{Q}^{{\cal M},{\mathcal E},l}})
\end{multline}
\begin{multline}
\label{Acc_decomposability_sym_err_case2}
    W_{t+1}(\mathbf{P_{11}^{\mathcal E}}, \mathbf{\Phi}(k+1,N), \mathbf{\overline{Q}^{{\cal M},{\mathcal E},l}},p_{01},\mathbf{\underline{Q}^{{\cal M},{\mathcal E},l}}) \\
    = p_{01} W_{t+1}(\mathbf{P_{11}^{\mathcal E}}, \mathbf{\Phi}(k+1,N), \mathbf{\overline{Q}^{{\cal M},{\mathcal E},l}},1,\mathbf{\underline{Q}^{{\cal M},{\mathcal E},l}}) \\
    + (1-p_{01}) W_{t+1}(\mathbf{P_{11}^{\mathcal E}}, \mathbf{\Phi}(k+1,N), \mathbf{\overline{Q}^{{\cal M},{\mathcal E},l}},0,\mathbf{\underline{Q}^{{\cal M},{\mathcal E},l}})
\end{multline}
\begin{multline}
\label{Acc_decomposability_sym_err_case3}
      W_{t+1}(\mathbf{P_{11}^{\mathcal E}}, \mathbf{\Phi}(k+1,N), \mathbf{\overline{Q}^{{\cal M},{\mathcal E},l}},p_{11}, \mathbf{\underline{Q}^{{\cal M},{\mathcal E},l}}) \\
      =  p_{11}  W_{t+1}(\mathbf{P_{11}^{\mathcal E}}, \mathbf{\Phi}(k+1,N), \mathbf{\overline{Q}^{{\cal M},{\mathcal E},l}},1, \mathbf{\underline{Q}^{{\cal M},{\mathcal E},l}}) \\
      + (1- p_{11}) W_{t+1}(\mathbf{P_{11}^{\mathcal E}}, \mathbf{\Phi}(k+1,N), \mathbf{\overline{Q}^{{\cal M},{\mathcal E},l}},0, \mathbf{\underline{Q}^{{\cal M},{\mathcal E},l}})
\end{multline}
Combing~\eqref{Acc_decomposability_sym_err_case1},~\eqref{Acc_decomposability_sym_err_case2},~\eqref{Acc_decomposability_sym_err_case3}, we have~\eqref{Acc_decomposability_sym_err}.
\end{itemize}
Combing the above analysis in two cases, we thus prove Lemma~\ref{lemma:decomposability_sym_err}.

\section{Proof of Lemma~\ref{lemma:monotonicity_sym_err}}

We proceed the proof by backward induction. Firstly, it is easy to verify that the lemma holds for slot $T$.

Assume that the lemma holds from slots $t+1,\cdots,T$, we now prove that it also holds for slot $t$ by distinguishing the following two cases.
\begin{itemize}
\item Case 1: channel $l$ is not sensed in slot $t$, i.e., $l\ge k+1$. In this case, the immediate reward is unrelated to $\omega_l$ and $\omega_l'$. Moreover, let $\Omega(t+1)$ and $\Omega'(t+1)$ denote the belief vector generated by $\Omega(t)=(\omega_1, \cdots, \omega_l, \cdots, \omega_N)$ and $\Omega'(t)=(\omega_1, \cdots, \omega_l', \cdots, \omega_N)$, respectively, it can be noticed that $\Omega(t+1)$ and $\Omega'(t+1)$ differ in only one element: $\omega_l'(t+1)\ge \omega_l(t+1)$. By induction, it holds that $W_{t+1}(\Omega'(t+1))\ge W_{t+1}(\Omega(t+1))$. Noticing~\eqref{eq:w_t_err_sym}, it follows that $W_{t}(\Omega'(t))\ge W_{t}(\Omega(t))$.
\item Case 2: channel $l$ is sensed in slot $t$, i.e., $l\le k$. Following Lemma~\ref{lemma:decomposability_sym_err} and after some straightforward algebraic operations, we have
\begin{multline*}
W_t(\omega_1, \cdots, \omega_l', \cdots, \omega_N)-W_t(\omega_1, \cdots, \omega_l, \cdots, \omega_N)= \\
(\omega'_l-\omega_l)[W_t(\omega_1, \cdots, 1, \cdots, \omega_N)-W_t(\omega_1, \cdots, 0, \cdots, \omega_N)].
\end{multline*}

Let ${\cal M}\triangleq {\cal N}(k)\setminus\{l\}=\{1, \cdots,l-1,l+1,\cdots,k\}$, by developing $W_t(\Omega(t))$ as a function of $\omega_l$, we have
\begin{align*}
W_t(\Omega(t)) &= F(\omega_1(t), \cdots, \omega_k(t))+\beta (1-\epsilon)\omega_l \sum_{{\cal E}\subseteq{\cal M}}Pr({\cal M}, {\cal E}) W_{t+1}(\Omega_{{\mathcal E}}(t+1)) \\
&+\beta [1-(1-\epsilon)\omega_l]\sum_{{\cal E}\subseteq{\cal M}}Pr({\cal M}, {\cal E}) W_{t+1}(\Omega_{{\mathcal E}}(t+1)).
\end{align*}

Let $\omega_l=0$ and $1$, respectively, we have
\begin{eqnarray*}
  W_t(\omega_1,\cdots,0,\cdots,\omega_n) &=& F(\omega_1,\cdots,0,\cdots,\omega_n)+\beta \sum_{{\cal E}\subseteq{\cal M}}Pr({\cal M}, {\cal E}) W_{t+1}(\Omega_0^{{\mathcal E}}(t+1)), \\
  W_t(\omega_1,\cdots,1,\cdots,\omega_n) &=& F(\omega_1,\cdots,1,\cdots,\omega_n)+\beta (1-\epsilon)\sum_{{\cal E}\subseteq{\cal M}}Pr({\cal M}, {\cal E}) W_{t+1}(\Omega_{1-\epsilon}^{{\mathcal E}}(t+1)) \\
  &+& \beta \epsilon\sum_{{\cal E}\subseteq{\cal M}}Pr({\cal M}, {\cal E}) W_{t+1}(\Omega_{\epsilon}^{{\mathcal E}}(t+1)),
\end{eqnarray*}
where
\begin{eqnarray*}
\Omega_0^{{\mathcal E}}(t+1)&=&(\mathbf{P_{11}^{\mathcal E}}, \mathbf{\Phi}(k+1,N), \mathbf{\overline{Q}^{{\cal M},{\mathcal E},l}},p_{01}, \mathbf{\underline{Q}^{{\cal M},{\mathcal E},l}}), \\
\Omega_{1-\epsilon}^{{\mathcal E}}(t+1)&=&(\mathbf{P_{11}^{\mathcal E}}, p_{11}, \mathbf{\Phi}(k+1,N), \mathbf{\overline{Q}^{{\cal M},{\mathcal E},l}}, \mathbf{\underline{Q}^{{\cal M},{\mathcal E},l}}), \\
\Omega_{\epsilon}^{{\mathcal E}}(t+1)&=&(\mathbf{P_{11}^{\mathcal E}}, \mathbf{\Phi}(k+1,N),\mathbf{\overline{Q}^{{\cal M},{\mathcal E},l}},p_{11}, \mathbf{\underline{Q}^{{\cal M},{\mathcal E},l}}).
\end{eqnarray*}

It can be checked that $\Omega_{1-\epsilon}^{{\mathcal E}}(t+1)\ge \Omega_0^{{\mathcal E}}(t+1)$ and $\Omega_{\epsilon}^{{\mathcal E}}(t+1)\ge \Omega_0^{{\mathcal E}}(t+1)$. It then follows from induction that given $\mathcal E$, $W_{t+1}(\Omega_{1-\epsilon}^{{\mathcal E}}(t+1))\ge W_{t+1}(\Omega_0^{{\mathcal E}}(t+1))$ and $W_{t+1}(\Omega_{1-\epsilon}^{{\mathcal E}}(t+1))\ge W_{t+1}(\Omega_0^{{\mathcal E}}(t+1))$. Noticing that $F$ is increasing, we then have
\begin{align*}
W_t(\omega_1,\cdots,1,\cdots,\omega_n)&-W_t(\omega_1,\cdots,0,\cdots,\omega_n) = F(\omega_1,\cdots,1,\cdots,\omega_n)-F(\omega_1,\cdots,0,\cdots,\omega_n) \\
&+\beta(1-\epsilon)\sum_{{\cal E}\subseteq{\cal M}}Pr({\cal M}, {\cal E}) [W_{t+1}(\Omega_{1-\epsilon}^{{\mathcal E}}(t+1))-W_{t+1}(\Omega_0^{{\mathcal E}}(t+1))] \\
&+\beta\epsilon\sum_{{\cal E}\subseteq{\cal M}}Pr({\cal M}, {\cal E}) [W_{t+1}(\Omega_{\epsilon}^{{\mathcal E}}(t+1))-W_{t+1}(\Omega_0^{{\mathcal E}}(t+1))] \ge 0.
\end{align*}
\end{itemize}

Combining the above analysis in two cases completes our proof.

\section{Proof of Lemma~\ref{lemma:exchange_err_sym} and Lemma~\ref{lemma:upper_bound_err_sym}}

Due to the dependency between the two lemmas, we prove them together by backward induction.

\textbf{We first show that Lemma~\ref{lemma:exchange_err_sym} and Lemma~\ref{lemma:upper_bound_err_sym} hold for slot $T$.} It is easy to verify that Lemma~\ref{lemma:exchange_err_sym} holds.

We then prove Lemma~\ref{lemma:upper_bound_err_sym}. Noticing that $p_{01}\le \omega_N\le \omega_k\le p_{11}\le 1$, we have
\begin{align*}
W_{T}(\omega_{1},\cdots,\omega_{N})&-W_{T}(\omega_{N},\omega_{1},\cdots,\omega_{N-1}) = F(\omega_{1},\cdots,\omega_{k}) - F(\omega_{N},\omega_{1},\cdots,\omega_{k-1}) \\
  &= (\omega_{k}-\omega_{N})[F(\omega_{1},\cdots,\omega_{k-1},1)-F(\omega_{1},\cdots,\omega_{k-1},0)] \le (1-\omega_N)\Delta_{max}, \\
W_{T}(\omega_{1},\cdots,\omega_{N})&-W_{T}(\omega_{N},\omega_{2},\cdots,\omega_{N-1}, \omega_1) = F(\omega_{1},\cdots,\omega_{k}) - F(\omega_{N},\omega_{2},\cdots,\omega_{k-1}) \\
  &= (\omega_{1}-\omega_{N})[F(1,\omega_{2},\cdots,\omega_{k})-F(0,\omega_{2},\cdots,\omega_{k})] \le (p_{11}-p_{01})\Delta_{max}.
\end{align*}
Lemma~\ref{lemma:upper_bound_err_sym} thus   holds for slot $T$.

\textbf{Assume that Lemma~\ref{lemma:exchange_err_sym} and Lemma~\ref{lemma:upper_bound_err_sym} hold for slots $T,\cdots, t+1$, we now prove that it holds for slot $t$.}

\textbf{We first prove Lemma~\ref{lemma:exchange_err_sym}.} We distinguish the following three cases considering $l<m$:
\begin{itemize}
\item Case 1: $l\ge k+1$. In this case, we have
\begin{align*}
W_t&(\omega_{1},\cdots,\omega_{l},\cdots,\omega_{m},\cdots,\omega_{N}) - W_{t}(\omega_{1},\cdots,\omega_{m},\cdots,\omega_{l},\cdots,\omega_{N}) \\
&= (\omega_{l}-\omega_{m})[W_{t}(\omega_{1},\cdots,1,\cdots, 0,\cdots,\omega_{N}) - W_{t}(\omega_{1},\cdots,0,\cdots,1,\cdots,\omega_{N})] \\
&= (\omega_{l}-\omega_{m})\beta \sum_{{\mathcal E}\subseteq{{\cal N}(k)}} Pr({\cal N}(k),{\mathcal E}) [W_{t+1}(\Omega_{{\mathcal E}}(t+1))-W_{t+1}(\Omega'_{{\mathcal E}}(t+1))],
\end{align*}
where
\begin{align*}
\Omega_{{\mathcal E}}(t+1)&=(\mathbf{P_{11}^{{\mathcal E}}},\tau(\omega_{k+1}),\cdots,p_{11},\cdots, p_{01},\cdots,\tau(\omega_{N}), \mathbf{Q^{{\cal N}(k),{\mathcal E}}}), \\
\Omega'_{{\mathcal E}}(t+1)&=(\mathbf{P_{11}^{{\mathcal E}}},\tau(\omega_{k+1}),\cdots,p_{01},\cdots,p_{11},\cdots,\tau(\omega_{N}), \mathbf{Q^{{\cal N}(k),{\mathcal E}}}).
\end{align*}
It follows from the induction result that $W_{t+1}(\Omega_{{\mathcal E}}(t+1))\ge W_{t+1}(\Omega'_{{\mathcal E}}(t+1))$. Hence
\begin{equation*}
    W_{t}(\omega_{1},\cdots,\omega_{l},\cdots,\omega_{m},\cdots,\omega_{N})\ge W_{t}(\omega_{1},\cdots,\omega_{m},\cdots,\omega_{l},\cdots,\omega_{N}).
\end{equation*}

\item Case 2: $l\le k$ and $m\ge k+1$. In this case, denote ${\cal M}\triangleq {\cal N}(k)\setminus\{l\}$, it can be noted that $ \mathbf{Q^{{\cal M},{\mathcal E}}}= \mathbf{\underline{Q}^{{\cal M},{\mathcal E},l}}+ \mathbf{\overline{Q}^{{\cal M},{\mathcal E},l}}$. In this case, we have
\begin{IEEEeqnarray*}{rCl}
W_t&(\omega_{1}&,\cdots,\omega_{l},\cdots,\omega_{m},\cdots,\omega_{N})- W_{t}(\omega_{1},\cdots,\omega_{m},\cdots,\omega_{l},\cdots,\omega_{N}) \\
& =& (\omega_{l}-\omega_{m})[W_{t}(\omega_{1},\cdots,1,\cdots, 0,\cdots,\omega_{N}) - W_{t}(\omega_{1},\cdots,0,\cdots,1,\cdots,\omega_{N})] \\
& = &(\omega_{l}-\omega_{m})[F(\omega_{1},\cdots,1,\cdots,\omega_{k})-F(\omega_{1},\cdots,0,\cdots,\omega_{k})+ \\
&& \beta \sum_{{\mathcal E}\subseteq{\cal M}} Pr({\cal M},{\mathcal E}) [(1-\epsilon)W_{t+1}(\mathbf{P_{11}^{{\mathcal E}}},p_{11},\tau({\omega_{k+1}}),\cdots,p_{01},\cdots,\tau({\omega_{N}}), \mathbf{Q^{{\cal M},{\mathcal E}}}) +\\
&& \epsilon W_{t+1}(\mathbf{P_{11}^{{\mathcal E}}},\tau({\omega_{k+1}}),\cdots,p_{01},\cdots,\tau({\omega_{N}}),\mathbf{\overline{Q}^{{\cal M},{\mathcal E},l}},p_{11}, \mathbf{\underline{Q}^{{\cal M},{\mathcal E},l}})- \\
&& W_{t+1}(\mathbf{P_{11}^{{\mathcal E}}},\tau({\omega_{k+1}}),\cdots,p_{11},\cdots,\tau({\omega_{N}}), \mathbf{\overline{Q}^{{\cal M},{\mathcal E},l}},p_{01},\mathbf{\underline{Q}^{{\cal M},{\mathcal E},l}})] \\
&\ge &(\omega_{l}-\omega_{m})[\Delta_{min}+ \beta \sum_{{\mathcal E}\subseteq{\cal M}} Pr({\cal M},{\mathcal E}) \cdot [(1-\epsilon)W_{t+1}(p_{01},\mathbf{P_{11}^{{\mathcal E}}},p_{11},\tau({\omega_{k+1}}),\cdots,\tau({\omega_{N}}), \mathbf{Q^{{\cal M},{\mathcal E}}})+ \\
&& \epsilon W_{t+1}(p_{01},\mathbf{P_{11}^{{\mathcal E}}},\tau({\omega_{k+1}}),\cdots,\tau({\omega_{N}}),
\mathbf{Q^{{\cal M},{\mathcal E}}},p_{11})-\\
&& W_{t+1}(\mathbf{P_{11}^{{\mathcal E}}},p_{11},\tau({\omega_{k+1}}),\cdots,\tau({\omega_{N}}), \mathbf{Q^{{\cal M},{\mathcal E}}},p_{01})] \\
&\ge & (\omega_{l}-\omega_{m})\left[\Delta_{min}-\beta\sum_{{\mathcal E}\subseteq{\cal M}} Pr({\cal M},{\mathcal E})\cdot\right. \nonumber \\
&& \left.\left((1-\epsilon)(1-p_{01})\Delta_{max}+\epsilon (p_{11}-p_{01})\Delta_{max}\frac{1-[\beta(1-\epsilon)(p_{11}-p_{01})]^{T-t}}{1-\beta(1-\epsilon)(p_{11}-p_{01})}\right)\right] \\
&\ge &(\omega_{l}-\omega_{m})\sum_{{\mathcal E}\subseteq{\cal M}} Pr({\cal M},{\mathcal E})\cdot \\
&&\left[\Delta_{min}-\beta\left((1-\epsilon)(1-p_{01})\Delta_{max}+\epsilon (p_{11}-p_{01})\Delta_{max}\frac{1}{1-(1-\epsilon)(p_{11}-p_{01})}\right)\right] \ge 0,
\end{IEEEeqnarray*}
where the first inequality follows the induction result of Lemma~\ref{lemma:exchange_err_sym}, the second inequality follows the induction result of Lemma~\ref{lemma:upper_bound_err_sym}, the third inequality follows the condition in the lemma.
\item Case 3: $l,m\ge k$. This case follows Lemma~\ref{lemma:symmetry_sym_err}.
\end{itemize}

Lemma 6 is thus proven for slot $t$.

\textbf{We then proceed to prove Lemma~\ref{lemma:upper_bound_err_sym}}. We start with the first inequality. We develop $W_t$ w.r.t. $\omega_k$ and $\omega_N$ according to Lemma~\ref{lemma:decomposability_sym_err} as follows:
\begin{IEEEeqnarray}{rCl}
&& W_{t}(\omega_{1},\cdots,\omega_{k-1},\omega_{k},\cdots,\omega_{n-1},\omega_{n})-W_{t}(\omega_{n},\omega_{1},\cdots,\omega_{k-1},\omega_{k},...,\omega_{n-1}) \nonumber \\
 &&= \omega_{k}\omega_{n}[W_{t}(\omega_{1},\cdots,\omega_{k-1},1,\omega_{k+1},\cdots,\omega_{n-1},1)-W_{t}(1,\omega_{1},\cdots,\omega_{k-1},1,\omega_{k+1},\cdots,\omega_{n-1})] \nonumber \\
 && + \omega_{k}(1-\omega_{n})[W_{t}(\omega_{1},\cdots,\omega_{k-1},1,\omega_{k+1},\cdots,\omega_{n-1},0)-W_{t}(0,\omega_{1},\cdots,\omega_{k-1},1,\omega_{k+1},\cdots,\omega_{n-1})]\nonumber \\
 && + (1-\omega_{k})\omega_{n}[W_{t}(\omega_{1},\cdots,\omega_{k-1},0,\omega_{k+1},\cdots,\omega_{n-1},1)-W_{t}(1,\omega_{1},\cdots,\omega_{k-1},0,\omega_{k+1},\cdots,\omega_{n-1})]\nonumber \\
 && + (1-\omega_{k})(1-\omega_{n})[W_{t}(\omega_{1},\cdots,\omega_{k-1},0,\omega_{k+1},\cdots,\omega_{n-1},0)-W_{t}(0,\omega_{1},\cdots,\omega_{k-1},0,\omega_{k+1},\cdots,\omega_{n-1})] \nonumber \\
\label{eq:4_term}
\end{IEEEeqnarray}

We proceed the proof by upbounding the four terms in~\eqref{eq:4_term}.

For the first term, we have
\begin{multline*}
W_{t}(\omega_{1},\cdots,\omega_{k-1},1,\omega_{k+1},\cdots,\omega_{n-1},1)-W_{t}(1,\omega_{1},\cdots,\omega_{k-1},1,\omega_{k+1},\cdots,\omega_{n-1})\\
   = \beta\sum_{{\mathcal E}\subseteq{\cal N}(k-1)} Pr({\cal N}(k-1),{\mathcal E})\cdot
 [(1-\epsilon)W_{t+1}(\mathbf{P_{11}^{{\mathcal E}}},p_{11},\mathbf{\Phi}(k+1,N-1), p_{11}, \mathbf{Q^{{\cal N}(k-1),{\mathcal E}}}) \\
 + \epsilon W_{t+1}(\mathbf{P_{11}^{{\mathcal E}}},\mathbf{\Phi}(k+1,N-1),p_{11},\mathbf{{Q}^{{\cal N}(k-1),{\mathcal E}}},p_{11}) \\
~~~~~~~~~~-(1-\epsilon)W_{t+1}(p_{11},\mathbf{P_{11}^{{\mathcal E}}},p_{11},\mathbf{\Phi}(k+1,N-1), \mathbf{Q^{{\cal N}(k-1),{\mathcal E}}}) \\
-\epsilon W_{t+1}(\mathbf{P_{11}^{{\mathcal E}}},p_{11},\mathbf{\Phi}(k+1,N-1),p_{11}, \mathbf{Q^{{\cal N}(k-1),{\mathcal E}}})] \le 0
\end{multline*}
where, the inequality follows the induction of Lemma~\ref{lemma:exchange_err_sym}.

For the second term, we have
\begin{IEEEeqnarray*}{rCl}
&&W_{t}(\omega_{1},\cdots,\omega_{k-1},1,\omega_{k+1},\cdots,\omega_{n-1},0)-W_{t}(0,\omega_{1},\cdots,\omega_{k-1},1,\omega_{k+1},\cdots,\omega_{n-1})\\ &=& F(\omega_{1},\cdots,\omega_{k-1},1)-F(0,\omega_{1},\cdots,\omega_{k-1})  \\
&&+\beta\sum_{{\mathcal E}\subseteq{\cal N}(k-1)} Pr({\cal N}(k-1),{\mathcal E})\cdot [
(1-\epsilon)W_{t+1}(\mathbf{P_{11}^{{\mathcal E}}},p_{11},\mathbf{\Phi}(k+1,N-1),p_{01}, \mathbf{Q^{{\cal N}(k-1),{\mathcal E}}}) \\
&& +\epsilon W_{t+1}(\mathbf{P_{11}^{{\mathcal E}}},\mathbf{\Phi}(k+1,N-1), p_{01},\mathbf{{Q}^{{\cal N}(k-1),{\mathcal E}}},p_{11}) -W_{t+1}(\mathbf{P_{11}^{{\mathcal E}}},p_{11},\mathbf{\Phi}(k+1,N-1),p_{01},\mathbf{Q^{{\cal N}(k-1),{\mathcal E}}})] \\
&=& F(\omega_{1},\cdots,\omega_{k-1},1)-F(0,\omega_{1},\cdots,\omega_{k-1}) +\beta\sum_{{\mathcal E}\subseteq{\cal N}(k-1)} Pr({\cal N}(k-1),{\mathcal E})\cdot \\
&& [\epsilon W_{t+1}(\mathbf{P_{11}^{{\mathcal E}}},\mathbf{\Phi}(k+1,N-1), p_{01},\mathbf{{Q}^{{\cal N}(k-1),{\mathcal E}}},p_{11}) - \epsilon W_{t+1}(\mathbf{P_{11}^{{\mathcal E}}},p_{11},\mathbf{\Phi}(k+1,N-1),p_{01},\mathbf{Q^{{\cal N}(k-1),{\mathcal E}}})] \\
&\le& \Delta_{max}
\end{IEEEeqnarray*}
following the induction of Lemma~\ref{lemma:exchange_err_sym}.

For the third term, we have
\begin{IEEEeqnarray*}{rCl}
 \IEEEeqnarraymulticol{3}{l} {
W_{t}(\omega_{1},\cdots,\omega_{k-1},0,\omega_{k+1},\cdots,\omega_{n-1},1)-W_{t}(1,\omega_{1},\cdots,\omega_{k-1},0,\omega_{k+1},\cdots,\omega_{n-1}) }\nonumber\\
&&=F(\omega_{1},\cdots,\omega_{k-1},0)-F(1,\omega_{1},\cdots,\omega_{k-1})+ \\
&&\beta\sum_{{\mathcal E}\subseteq{\cal N}(k-1)} Pr({\cal N}(k-1),{\mathcal E})[
W_{t+1}(\mathbf{P_{11}^{{\mathcal E}}},\mathbf{\Phi}(k+1,N-1), p_{11},\mathbf{Q^{{\cal N}(k-1),{\mathcal E}}}, p_{01})- \\
&& (1-\epsilon)W_{t+1}(p_{11}, \mathbf{P_{11}^{{\mathcal E}}},p_{01},\mathbf{\Phi}(k+1,N-1), \mathbf{Q^{{\cal N}(k-1),{\mathcal E}}})
- \epsilon W_{t+1}(\mathbf{P_{11}^{{\mathcal E}}},p_{01},\mathbf{\Phi}(k+1,N-1),p_{11}, \mathbf{Q^{{\cal N}(k-1),{\mathcal E}}})] \\
&\le& -\Delta_{min}+\beta\sum_{{\mathcal E}\subseteq{\cal N}(k-1)} Pr({\cal N}(k-1),{\mathcal E})[
W_{t+1}(\mathbf{P_{11}^{{\mathcal E}}},p_{11},\mathbf{\Phi}(k+1,N-1), \mathbf{Q^{{\cal N}(k-1),{\mathcal E}}}, p_{01})- \\
&& (1-\epsilon)W_{t+1}(p_{01},p_{11}, \mathbf{P_{11}^{{\mathcal E}}},\mathbf{\Phi}(k+1,N-1), \mathbf{Q^{{\cal N}(k-1),{\mathcal E}}})-
\epsilon W_{t+1}(p_{01},\mathbf{P_{11}^{{\mathcal E}}},\mathbf{\Phi}(k+1,N-1),\mathbf{Q^{{\cal N}(k-1),{\mathcal E}}},p_{11})] \\
&\le&-\Delta_{min} +  \beta\sum_{{\mathcal E}\subseteq{\cal N}(k-1)} Pr({\cal N}(k-1),{\mathcal E})\left[(1-\epsilon)(1-p_{01}) \Delta_{max}+\epsilon(p_{11}-p_{01})\Delta_{max}\frac{1-[\beta(1-\epsilon)(p_{11}-p_{01})]^{T-t}}{1-\beta(1-\epsilon)(p_{11}-p_{01})}\right] \\
&\le & \sum_{{\mathcal E}\subseteq{\cal N}(k-1)} Pr({\cal N}(k-1),{\mathcal E})\left[-\Delta_{min} + \beta \left[(1- \epsilon)(1-p_{01}) \Delta_{max} +  \epsilon(p_{11}-p_{01})\Delta_{max}\frac{1}{1-(1-\epsilon)(p_{11}-p_{01})}\right]\right] \le 0
\end{IEEEeqnarray*}
where the first inequality follows the induction result of Lemma~\ref{lemma:exchange_err_sym}, the second equality follows the induction result of Lemma~\ref{lemma:upper_bound_err_sym}, the forth inequality is due the condition in Lemma~\ref{lemma:upper_bound_err_sym}.

For the fourth term, we have
\begin{IEEEeqnarray*}{rCl}
 \IEEEeqnarraymulticol{3}{l} {
W_{t}(\omega_{1},\cdots,\omega_{k-1},0,\omega_{k+1},\cdots,\omega_{n-1},0)-W_{t}(0,\omega_{1},\cdots,\omega_{k-1},0,\omega_{k+1},\cdots,\omega_{n-1}) } \nonumber\\
&& = \beta\sum_{{\mathcal E}\subseteq{\cal N}(k-1)} Pr({\cal N}(k-1),{\mathcal E})  [W_{t+1}(\mathbf{P_{11}^{{\mathcal E}}},\mathbf{\Phi}(k+1,N-1),p_{01}, \mathbf{Q^{{\cal N}(k-1),{\mathcal E}}}, p_{01}) \nonumber \\
&& - W_{t+1}(\mathbf{P_{11}^{{\mathcal E}}}, p_{01},\mathbf{\Phi}(k+1,N-1),\mathbf{Q^{{\cal N}(k-1),{\mathcal E}}}, p_{01})] \nonumber \\
&&= \beta\sum_{{\mathcal E}\subseteq{\cal N}(k-1)} Pr({\cal N}(k-1),{\mathcal E})[W_{t+1}(\mathbf{P_{11}^{{\mathcal E}}},\mathbf{\Phi}(k+1,N-1),p_{01}, \mathbf{Q^{{\cal N}(k-1),{\mathcal E}}},p_{01}) \\
&& - W_{t+1}(p_{01},\mathbf{P_{11}^{{\mathcal E}}}, \mathbf{\Phi}(k+1,N-1),\mathbf{Q^{{\cal N}(k-1),{\mathcal E}}}, p_{01})] \nonumber \\
&&\leq \beta\sum_{{\mathcal E}\subseteq{\cal N}(k-1)} Pr({\cal N}(k-1),{\mathcal E})[W_{t+1}(\mathbf{P_{11}^{{\mathcal E}}},\mathbf{\Phi}(k+1,N-1), \mathbf{Q^{{\cal N}(k-1),{\mathcal E}}},p_{01},p_{01}) \nonumber \\
&& - W_{t+1}(p_{01},\mathbf{P_{11}^{{\mathcal E}}}, \mathbf{\Phi}(k+1,N-1),\mathbf{Q^{{\cal N}(k-1),{\mathcal E}}}, p_{01})] \nonumber \\
&&\le (1-p_{01}) \beta \Delta_{max}
\end{IEEEeqnarray*}
where, the second equality follows Lemma~\ref{lemma:symmetry_sym_err}, the first inequality follows the induction result of Lemma~\ref{lemma:exchange_err_sym} and the second inequality follows the induction result of Lemma~\ref{lemma:upper_bound_err_sym}.

Combing the above results of the four terms, we have
\begin{multline*}
    W_{t}(\omega_{1},\cdots,\omega_{N})-W_{t}(\omega_{n},\omega_{1},\cdots,\omega_{N-1})\\
\le \omega_{k}(1-\omega_{N}) \cdot \Delta_{max} + (1-\omega_{k})(1-\omega_{N}) \cdot (1-p_{01})\beta \Delta_{max}\\
\le \omega_{k}(1-\omega_{N})\Delta_{max} + (1-\omega_{k})(1-\omega_{N})\Delta_{max} \le (1-\omega_{N})\Delta_{max},
\end{multline*}
which completes the proof of the first part of Lemma~\ref{lemma:upper_bound_err_sym}.

Finally, we prove the second part of Lemma~\ref{lemma:upper_bound_err_sym}. To this end, denote ${\cal M}\triangleq \{2,\cdots,k\}$, we have
\begin{IEEEeqnarray*}{rCl}
\IEEEeqnarraymulticol{3}{l}{
W_{t}(\omega_{1},\cdots,\omega_{N})-W_{t}(\omega_{N},\omega_{2},\cdots,\omega_{N-1}, \omega_{1}) }  \\
&=& (\omega_{1}-\omega_{N})[W_{t}(1,\omega_{2},\cdots,\omega_{N-1},0)-W_{t}(0,\omega_{2},\cdots,\omega_{N-1},1)] \\
&=& (\omega_{1}-\omega_{N}) (F(1,\omega_{2},\cdots,\omega_{k})-F(0,\omega_{2},\cdots,\omega_{k}) + \beta\sum_{{\mathcal E}\subseteq{\cal M}} Pr({\cal M},{\mathcal E}) \cdot \\
&& [(1- \epsilon)W_{t+1}(\mathbf{P_{11}^{{\mathcal E}}},p_{11},\mathbf{\Phi}(k+1,N-1),p_{01}, \mathbf{Q^{{\cal M},{\mathcal E}}})+
   \epsilon W_{t+1}(\mathbf{P_{11}^{{\mathcal E}}},\mathbf{\Phi}(k+1,N-1),p_{01},p_{11}, \mathbf{Q^{{\cal M},{\mathcal E}}}) \\
&& -W_{t+1}(\mathbf{P_{11}^{{\mathcal E}}},\mathbf{\Phi}(k+1,N-1),p_{11}, p_{01}, \mathbf{Q^{{\cal M},{\mathcal E}}})] \\
&\le& (\omega_{1}-\omega_{N}) (\Delta_{max}+\beta\sum_{{\mathcal E}\subseteq{\cal M}} Pr({\cal M},{\mathcal E}) [(1- \epsilon)W_{t+1}(\mathbf{P_{11}^{{\mathcal E}}},p_{11},\mathbf{\Phi}(k+1,N-1),p_{01}, \mathbf{Q^{{\cal M},{\mathcal E}}}) \\
&&+\epsilon W_{t+1}(\mathbf{P_{11}^{{\mathcal E}}},\mathbf{\Phi}(k+1,N-1),p_{01},p_{11}, \mathbf{Q^{{\cal M},{\mathcal E}}})-W_{t+1}(\mathbf{P_{11}^{{\mathcal E}}},\mathbf{\Phi}(k+1,N-1),p_{01}, p_{11}, \mathbf{Q^{{\cal M},{\mathcal E}}})]) \\
&=& (\omega_{1}-\omega_{N}) (\Delta_{max} + \beta\sum_{{\mathcal E}\subseteq{\cal M}} Pr({\cal M},{\mathcal E}) [(1- \epsilon)W_{t+1}(\mathbf{P_{11}^{{\mathcal E}}},p_{11},\mathbf{\Phi}(k+1,N-1),p_{01},\mathbf{Q^{{\cal M},{\mathcal E}}})\\
&& -(1-\epsilon) W_{t+1}(\mathbf{P_{11}^{{\mathcal E}}},\mathbf{\Phi}(k+1,N-1),p_{01},p_{11}, \mathbf{Q^{{\cal M},{\mathcal E}}}]) \\
&\le& (\omega_{1}-\omega_{N}) (\Delta_{max} + \beta\sum_{{\mathcal E}\subseteq{\cal M}} Pr({\cal M},{\mathcal E}) [(1- \epsilon)W_{t+1}(\mathbf{P_{11}^{{\mathcal E}}},p_{11},\mathbf{\Phi}(k+1,N-1),\mathbf{Q^{{\cal M},{\mathcal E}}},p_{01})-\\
&& (1-\epsilon) W_{t+1}(p_{01}, \mathbf{P_{11}^{{\mathcal E}}},\mathbf{\Phi}(k+1,N-1), \mathbf{Q^{{\cal M},{\mathcal E}}},p_{11})]) \\
&\le& (p_{11}-p_{01})\left[\Delta_{max}+ \beta\sum_{{\mathcal E}\subseteq{\cal M}} Pr({\cal M},{\mathcal E})(1-\epsilon)\frac{1-[\beta(1-\epsilon)(p_{11}-p_{01})]^{T-t}}{1-\beta(1-\epsilon)(p_{11}-p_{01})} (p_{11}-p_{01})\Delta_{max}\right] \\
&=& \frac{1-[\beta(1-\epsilon)(p_{11}-p_{01})]^{T-t+1}}{1-\beta(1-\epsilon)(p_{11}-p_{01})} (p_{11}-p_{01})\Delta_{max}
\end{IEEEeqnarray*}
where the first two inequalities follows the recursive application of the induction result of Lemma~\ref{lemma:exchange_err_sym}, the third inequality follows the induction result of Lemma~\ref{lemma:upper_bound_err_sym}.

We thus complete the whole process of proving Lemma~\ref{lemma:exchange_err_sym} and Lemma~\ref{lemma:upper_bound_err_sym}.

\bibliographystyle{unsrt}
\bibliography{SenseError}

\end{document}